\documentclass[a4paper,11pt]{amsproc} 
\usepackage{amssymb}
\usepackage{amsthm}
\usepackage{epsf} 
\usepackage{bm}
\usepackage[dvipsnames]{xcolor}
\usepackage[T1]{fontenc}   
\usepackage{tikz}         
\makeindex
\newtheorem{thm}{Theorem}[section]   
\newtheorem{cor}[thm]{Corollary}
\newtheorem{question}[thm]{Question}

\newtheorem{lemma}[thm]{Lemma}
\newtheorem{prop}[thm]{Proposition}
\newtheorem{example}[thm]{Example}
\newtheorem{defn}[thm]{Definition}

\newtheorem{rem}[thm]{Remark}

\def\supp{\operatorname{Supp}}

\def\min{\operatorname{min}}

\def\max{\operatorname{max}}

\def\c1{\operatorname{c_1}}
\def\c2{\operatorname{c_2}}

\def\PP{{\mathbf P}}

\def\A{{\mathcal A}}
\def\C{{\mathcal C}}

\def\D{{\mathcal D}}

\def\+{\oplus}                   
\def\*{\otimes}                  

\def\Fq{\mathbb{F}_{q}}

\def\AA{\mathbb{A}}
\def\PP{\mathbb{P}}
\def\SS{\Sigma}
\def\P{\mathcal{P}}
\def\a{\alpha}
\def\Fq{\mathbb{F}_q}

\def\ev{\mathrm{ev}}

\begin{document}
\title[Codes from symmetric polynomials]{Codes from symmetric polynomials}

\author[Datta]{Mrinmoy Datta}
\address{Department of Mathematics 
 Indian Institute of Technology Hyderabad  \newline \indent 
}
\email{mrinmoy.datta@math.iith.ac.in} 
\thanks{}

\author[Johnsen]{Trygve Johnsen}
\address{Department of Mathematics and Statistics, 
 UiT-The Arctic University of Norway  
N-9037 Troms{\o}, Norway}
\email{trygve.johnsen@uit.no}
\thanks{The first named author is partially supported by a seed grant from Indian Institute of Technology Hyderabad. Both authors have been partially supported by grant 280731 from the Research Council of Norway,  and  by the project ``Pure Mathematics in Norway" through the Trond Mohn Foundation and Troms{\o} Research Foundation.}

\date{\today}

\begin{abstract}
We define and study a class of Reed-Muller type error-correcting codes obtained from elementary symmetric functions in finitely many variables. We determine the code parameters and higher weight spectra in the simplest cases.
\end{abstract}

\maketitle

\section{Introduction}

Over the last decades, good examples of error-correcting codes have been constructed  using algebraic geometric techniques. The codes constructed this way are linear codes over a given finite field $k$, where each member of a finite dimensional vector space of functions, say $V$,  are evaluated at a finite set of points, say $S$, all lying in an affine space or a projective space over the same field $k$. Examples are simplex codes,  Reed-Muller codes (\cite{W}), algebraic-geometric codes with $S$ a curve (Goppa codes) (\cite{G}) or a higher dimensional variety (\cite{BDH}), Grassmann codes (\cite{N}), and codes where the points in question  represent 
synmmetric or skew-symmetric matrices (\cite{BS}).

Having defined such codes, it is imperative that one looks for their parameters such as dimensions, minimum distance, weight distributions, generalized Hamming weights etc. These questions are often related to question that are interesting from the perspective of algebraic geometry, number theory and various branches of discrete mathematics. For instance, one checks easily that the minimum distance of a code defined using methods described above is equivalent to determining the maximum possible number of zeroes that a function in $V$ (that does not vanish identically in $S$) may have in $S$. 

In this paper we study a class of codes that are motivated from the Reed-Muller codes. While defining a Reed-Muller code, one evaluates the set of all reduced polynomials of degrees bounded above by a given quantity on the whole of affine space. Instead, here we consider a subspace of the set of all symmetric polynomials and evaluate them on points from affine spaces that have pairwise distinct coordinates. As it turns out, the relative minimum distance of our codes is same as that of Reed-Muller codes. However, relative dimension of the code is not as good. To this end, we introduce a modified family of codes that has the same relative minimum distance, but a better rate. We also show that, like the Reed-Muller codes, the new codes are also generated by minimum weight codewords. This property, in particular, makes the duals of the new codes useful. 


This article is organized as follows: In Section 2,  we study the number of points over finite fields with pairwise distinct coordinates satisfying multivariate symmetric polynomials that are linear combinations of elementary symmetric polynomials over a finite field. In Section 3, we introduce the new family of codes and study their properties, such as their dimension, minimum weight and minimum weight codewords. In Section 4, we derive upper bounds on the generalized Hamming weights of the codes. In Section 5, we work with the codes that occur from symmetric polynomials in two variables and prove several results including their generalized Hamming weights, weight distributions and higher weight spectra. In Section 6, we specially concentrate on trivariate symmetric polynomials over a field with $5$ elements for the sake of illustrating the difficulties in obtaining the parameters in higher dimensions.

\section{Symmetric polynomials and their distinguished zeroes} \label{defs}
Let $k$ be a field. In most cases, we shall restrict our attention to the case when $k = \Fq$, i.e. $k$ is a finite field with $q$ elements where $q$ is a prime power. For a positive integer $m$ and a nonnegative integer $i$, we denote by $\sigma_m^i$ the $i$-th elementary symmetric polynomial in $m$ variables $x_1, \dots, x_m$, i.e.,
$$\sigma_{m}^i = \sum_{1 \le j_1 < \cdots < j_i \le m} x_{j_1} \cdots x_{j_i}$$
for $1 \le i \le m$ and $\sigma_m^0 = 1$. 
 It is well known that any symmetric polynomial $f \in k [x_1, \dots, x_m]$ can be written as an algebraic expression in $\sigma_m^0, \dots, \sigma_m^m$. However, in this article we are interested in symmetric polynomials that are $k$-linear combinations of elementary symmetric polynomials. We denote by $\SS_m$ the $k$-linear subspace generated by the elementary symmetric polynomials $\sigma_m^0, \dots, \sigma_m^m$. Note that $\dim_k \SS_m = m+1$.

For a given polynomial $f \in k [x_1, \dots, x_m]$, we denote by $Z_k (f)$ the set of zeroes of $f$ in $\AA^m (k)$, the $m$-dimensional affine space over $k$. A point $(a_1, \dots, a_m) \in \AA^m (k)$ is said to be \textit{distinguished} if $a_i \neq a_j$ whenever $i \neq j$. In this paper, we are interested in the distinguished zeroes of symmetric polynomials described in the last paragraph. For ease of reference, we shall denote by $\AA_D (k)^m$ the set of all distinguished points of $\AA^m (k)$. For a subset $S \subset k$, and a polynomial $f \in k[x_1, \dots, x_m]$, we denote by $Z_{S, D}(f)$ the set of all distinguished zeroes of $f$ in $S^m$. Thus,
$$Z_{S, D}(f) := \{ (a_1, \dots, a_m) \in S^m \mid f(a_1, \dots, a_m) = 0, a_i \neq a_j \ \text{for all} \  i \neq j\}.$$
In particular, given a polynomial $f \in k [x_1, \dots, x_m]$, we denote by $Z_{k, D} (f)$ the set of distinguished zeroes of $f$ in $\AA^m (k)$. 

Next we introduce a combinatorial notation for ease of reading. For positive integers $n, r$ we denote by $\P (n, r)$ the number of possible arrangements of $r$ objects taken from $n$ distinct objects. More precisely,
$$
\P(n,r)=
\begin{cases}
{n \choose r} r! \ \ \ \text{if} \ \ r \le n \\
0 \ \ \ \ \ \text{otherwise}.
\end{cases}
$$
It follows trivially that $|\AA_D (\Fq)^m| = \P (q, m)$. We are interested in analyzing the number of distinguished zeroes of a symmetric polynomial that is a linear combinations of the elementary symmetric polynomials on certain finite grids in $\AA^m (k)$. Before we state our main result towards this direction, let us state a few remarks on such polynomials. Let $f \in k[x_1, \dots, x_m]$ be given by
\begin{equation}\label{eff}
f = a_0 + a_1 \sigma_m^1 + \dots + a_m \sigma_m^m
\end{equation}
where $a_0, \dots, a_m \in k$. It can be verified readily that 
\begin{equation}\label{down}
f = \left( a_0 + a_1 \sigma_{m-1}^1 + \dots + a_{m-1} \sigma_{m-1}^{m-1} \right) + x_m  \left( a_1 + a_2 \sigma_{m-1}^1 + \dots + a_{m} \sigma_{m-1}^{m-1} \right).
\end{equation}
For simplicity, we shall write 
\begin{equation}\label{onev}
f = f_1 + x_m f_2,
\end{equation}
where 
$$f_1 = a_0 + a_1 \sigma_{m-1}^1 + \dots + a_{m-1} \sigma_{m-1}^{m-1} \ \ \text{and} \ \ f_2= a_1 + a_2 \sigma_{m-1}^1 + \dots + a_{m} \sigma_{m-1}^{m-1}.$$

We may readily observe that a polynomial $f$ as in equation \eqref{eff} can be classified in two types:

\textbf{Type I: $f_1$ and $f_2$ are linearly dependent.} In this case, there exists $\a \in k$ such that 
$$a_ i = \a a_{i+1} \ \ \ \text{for\ all} \ \  i = 0, \dots, m-1.$$
If $a_m = 0$, then $f$ is a constant polynomial. On the other hand, if $a_m \neq 0$, then 
$$f = a_m (\a^m + \a^{m-1} \sigma_m^1 + \dots + \sigma_m^m).$$
As a consequence, if $f$ is of Type I, then $f = a_m \displaystyle{\prod_{i=1}^m} (\a + x_i)$. 

\textbf{Type II: $f_1$ and $f_2$ are linearly independent.} It is not hard to verify that in this case $f$ is absolutely irreducible, i.e. $f$ is irreducible in an algebraic closure of $k$. 

Note that, if we identify a nonzero polynomial as in \eqref{eff} with the point $[a_0 : a_1 : \dots : a_m]$ in a projective space $\PP^m (k)$ of dimension $m$ over the field $k$, then the polynomials of Type I correspond to (upto multiplication by a nonzero element of $k$) some $k$-rational points of the rational normal curve in $\PP^m (\bar{k})$.  Indeed,  the $k$-rational points of the rational normal curves are of the form $[\a^m : \cdots: \a : 1]$ for some $\a \in k$ or $[0: \cdots :0 : 1]$.  Under the correspondence mentioned as above,  we see that the polynomials of type I (upto multiplication by a nonzero constant in $k$) correspond to the $k$-rational points (of the first kind as mentioned above) of the rational normal curve in $\PP^m (\bar{k})$.  We are now ready to state the first main result of this article. 
%

\begin{thm} \label{basic}
Let $m$ be a positive integer and $S$ be a finite subset of $k$ with $|S| \ge m$. If $f$ is a nonzero symmetric polynomial as in \eqref{eff}, then
\begin{equation}\label{bound}
|Z_{S, D} (f)| \le m \P (|S| - 1, m-1).
\end{equation}
This bound is attained if and only if $f$ is a nonconstant Type I polynomial given by 
$$f = c \prod_{i=1}^m (x_i - b)$$
for some $c \in k$ and $b \in S$. 
Moreover, if $f$ is  non-zero and not of the above type, then
\begin{equation}\label{bound2}
|Z_{S, D}(f)| \le m \P(|S|-1, m-1)) - (|S|-m) \P (|S|-2, m-2).
\end{equation}
\end{thm}

\begin{proof}
We prove the inequality \eqref{bound} by induction on $|S|$. 
Suppose that $|S| = 1$. Then $m = 1$ and the assertion follows trivially. Suppose that the assertion is true for all $T \subset k^m$ where $|T| < |S|$ and $m \le |T|$. We distinguish two cases:

\textbf{Case 1:  $f$ is of type I.} In this case, we may write $$f=c(x_1-b)(x_2-b)\cdots(x_m-b)$$ for some $b \in k$.
Note that $(a_1,\cdots,a_m) \in Z_{S, D} (f)$ if and only if $b \in S$ and $a_i=b$ for some $i$. Consequently, 
$$
| Z_{S, D} (f)|=
\begin{cases}
m \P (|S| - 1, m-1) \ \ \ \text{if} \ \  b \in S \\
0 \ \ \ \ \ \text{otherwise}.
\end{cases}
$$

\textbf{Case II: $f$ is of type II.} Write $f = f_1 + x_m f_2$ as in equation \eqref{onev}. Since $f_1$ and $f_2$ are linearly independent, for every $\a \in k$, the polynomial $f (x_1, \dots, x_{m-1}, \a)$ is a nonzero symmetric polynomial that is a linear combination of the elementary symmetric polynomials in $m-1$ variables. 
 Using the inductive hypothesis, we obtain,
 \begin{align*}
&|Z_{S, D} (f)|  \\
&= \sum_{\a \in S} |Z_{S \setminus \{\a\}, D} (f (x_1, \dots, x_{m-1}, \a))|  \\
&\le |S| (m-1) \P (|S| - 2, m-2) \\
&= (|S| - 1) (m-1) \P(|S| - 2, m-2) + (m-1) \P(|S| - 2, m-2) \\
&= (m-1) \P(|S| - 1, m-1) + (m-1) \P(|S| - 2, m-2) \\
&= m   \P(|S| - 1, m-1)  -  (\P(|S| - 1, m-1) + (m-1) \P(|S| - 2, m -2)) \\
& = m   \P(|S| - 1, m-1)  - (|S| - m) \P(|S| - 2, m -2)).
\end{align*}
This completes the proof. 
\end{proof}

We now apply the result to the particular case when $S = \Fq$ to get the following corollary.

\begin{cor}\label{cor1}
Let $f \in \Fq [x_1, \dots, x_m]$ be as in \eqref{eff}. If $m \le q$ and $f \neq 0$, then $|Z_{\Fq, D} (f)| \le m \P(q - 1, m - 1)$. Moreover, the equality holds if and only if $f$ is of Type I.
\end{cor}


Having known the maximum number of distinguished zeroes of a polynomial as in equation \eqref{eff}, it is important to address the following question.

\begin{question}\label{q1}
Given $f \in \Fq [x_1, \dots, x_m]$ as in \eqref{eff}, what is the possible number of distinguished zeroes in $\AA^m (\Fq)$ that $f$ may admit? 
\end{question} 

One can readily note that $|Z_{\Fq, D}(f)|$ is always divisible by $m!$. Furthermore, if $f$ is a nonzero constant polynomial, then it has no zeroes. If $f$ is a zero polynomial then it has $\P (q, m)$ distinguished zeroes. Moreover, thanks to Corollary \ref{cor1}, if $f$ is nonzero and of Type I, then it has $m \P(q-1, m-1)$ distinguished zeroes. We remark that the above question is equivalent to the question of determination of the weight distribution of the code defined in Section \ref{2}. In general, it is a hard question to answer. Here we completely work out the case when $m = 2$ and leave the general question open for further research.  

\begin{thm}\label{qodd} 
Let $q$ be odd, $m=2$, and $f \in \Fq[x_1, x_2]$ be given by $f = a_0 + a_1 (x_1 + x_2) + a_2 x_1 x_2$. If $\D :=  a_1^2 - a_0 a_2$, then.
\begin{equation*}
|Z_{\Fq, D} (f)|=
\begin{cases}
0,  \ \ \ \text{if} \ a_0 \neq 0  \ \ \text{and}\ \ (a_1, a_2) = (0, 0) \\
q-3, \ \  \text{if} \   a_2 \neq 0, \D \in \ \Fq^2  \ \text{and} \ \D \neq 0 \\
q-1, \ \  \text{if} \  \{a_2 = 0 \  \text{and} \  a_1 \neq 0\} \ \text{or} \\
\ \ \ \ \ \ \ \ \ \{a_2 \neq 0, \D \not\in \ \Fq^2  \ \text{and} \ \D \neq 0\} \\
2(q-1), \ \ \text{if} \ a_2 \neq 0 \ \text{and} \ \D = 0 \\
q(q-1), \ \ \text{if} \ (a_0, a_ 1, a_2) = (0, 0, 0)
\end{cases}
\end{equation*}
\end{thm}

\begin{proof}
If $f = 0$, then $|Z_{\Fq, D}(f)| = \P(q, 2)$. Conversely, it is clear from Corollary \ref{cor1} that if $|Z_{\Fq, D}(f)| = \P(q, 2)$, then $f=0$.  So we may assume that $f \neq 0$, i.e. $(a_0, a_1, a_2) \neq (0, 0, 0)$.   We distinguish the proof into several cases:

\begin{enumerate}
\item[\textbf{Case 1:}] Suppose $a_2 = 0$. If $a_1 =0$, then $f$ is a nonzero constant polynomial which does not have any zeroes. So we may assume that $a_1 \neq 0$. Then the polynomial $ a_0 + a_1 (x_1 + x_2)$ has $q - 1$ distinguished zeroes. 

\item[\textbf{Case 2:}] Suppose $a_2 \neq 0$. We may write
\begin{align*}
f (x_1, x_2) &= a_2 x_1 x_2 + a_1(x_1 + x_2) + a_0 \\
&= a_2 \left(x_1 x_2 + \frac{a_1}{a_2} (x_1 + x_2) + \frac{a_1^2}{a_2^2} \right) + {a_0} - \frac{a_1^2}{a_2} \\
&= a_2 \left( x_1 + \frac{a_1}{a_2}\right) \left( x_2 + \frac{a_1}{a_2}\right) + {a_0} - \frac{a_1^2}{a_2}
\end{align*}
By using the change of coordinates $X_1 = x_1 + a_1/a_2$ and $X_2 = x_2 + a_1/a_2$, and we get a new polynomial 
$$f' (X_1, X_2) =  X_1 X_2 - \frac{a_0a_2 - a_1^2}{a_2^2}.$$
It is clear that there is a one-one correspondence between the set of distinguished zeroes of $f$ and $f'$. This leads us to analyzing the distinguished zeroes of the polynomial $f'$. Note that the number of distinguished zeroes of $f'$ depends of the quantity $\D$.

\begin{enumerate}
\item[\textit{Subcase 1:}] Suppose $\D = 0$. Then the polynomial $f' (X_1, X_2) = X_1X_2$ has exactly $2(q-1)$ many distinguished zeroes. 

\item[\textit{Subcase 2:}] Suppose $\D \neq 0$ and $\D$ is a square in $\Fq$. Note that the polynomial $X_1 X_2 - \D/a_2^2$ has $q - 1$ zeroes.  Since $\D$ is a square in $\Fq$,   the polynomial $x^2 - \D/a_2^2$ has two distinct roots in $\Fq$.  As a consequence the polynomial $X_1X_2 - \D / a_2^2$ has exactly two zeroes that are nondistinguished, namely $(c_1, c_1)$ and $(c_2, c_2)$, where $c_1, c_2$ are the distinct solutions to the equation $x^2 - \D/a_2^2 = 0$. Thus,  polynomial $X_1 X_2 - \D/a_2^2$ have exactly $q-3$ distinguished zeroes.

\item[\textit{Subcase 3:}] Suppose $\D \neq 0$ and $\D$ is not a square in $\Fq$. In this case, all the zeroes of $X_1 X_2 - \D/a_2^2$ are distinguished.  This follows from the fact that the one variable polynomial equation $x^2 - \D/a_2^2 = 0$ has no solutions in $\Fq$. As a consequence, the number of distinguished zeroes of such a polynomial is $q-1$.
\end{enumerate}
\end{enumerate}
This completes the proof. 
\end{proof}

\begin{rem}\label{rem1}\normalfont
It is not very difficult to count the number of polynomials that have $0, q-3, q-1, 2(q-1),$ and $q(q-1)$ distinguished zeroes. It is trivial to see that there are $q-1$ nonzero constant polynomials admitting no zeroes and exactly one polynomial, namely the zero polynomial, admitting $q(q-1)$ distinguished zeroes. In order to count the number of polynomials $a_0 + a_1(x_1 + x_2) + a_2 x_1 x_2$, or equivalently, the tuples $(a_0, a_1, a_2)$ satisfying the  conditions $a_2 \neq 0$ and $\D$ is a nonzero square in $\Fq$, we note that there are $(q-1)/2$ possible values for $\D$, and for each of these choices, the $q(q-1)$ choices of $(a_1, a_2)$ (namely $q-1$ choices for a nonzero value of $a_2$ and $q$ choices for $a_1$) determines $a_0$ uniquely.  This results in a total of $q(q-1)^2/2$ many polynomials admitting $q-3$ distinguished zeroes.  Furthermore,  from Theorem \ref{qodd} we note that a  polynomial $a_0 + a_1 (x_1 + x_2) + a_2x_1x_2$ has $q-1$ distinguished zeros iff 
$$(a_0, a_1, a_2) \in \{a_2 = 0 \  \text{and} \  a_1 \neq 0\} \cup \{a_2 \neq 0, \D \not\in \ \Fq^2  \ \text{and} \ \D \neq 0\}.$$
Let $S_1 = \{a_2 = 0 \  \text{and} \  a_1 \neq 0\}$ and $S_2 = \{a_2 \neq 0, \D \not\in \ \Fq^2  \ \text{and} \ \D \neq 0\}$. Clearly, $|S_1| = q (q - 1)$.  To enumerate $S_2$,  we note that the quantity $\D$ takes $(q-1)/2$ possible values.  Now for each of these $(q-1)/2$ values of $\D$,  the element $a_0$ is determined uniquely for each choice of $(a_1, a_2)$ satisfying $a_2 \neq 0$.  It follows that $|S_2| = q(q-1)(q-1)/2$.  Thus,
$$|S_1| + |S_2| = q(q-1) + q(q-1)\frac{(q-1)}{2} = \frac{q(q-1)(q+1)}{2}.$$
Since $S_1$ and $S_2$ are disjoint, we see that there are exactly $q(q-1)(q+1)/2$ many polynomials having $q-1$ disinguished zeroes.  Finally, the polynomials that have exactly $2(q-1)$ distinguished zeroes are polynomials of Type I and there are $q(q-1)$ many polynomials of this type.  Hence there are exactly $q(q-1)$ many polynomials that admit $2(q-1)$ distinguished zeroes. We tabulate this data in Table \ref{tab1} below.

\end{rem}
\begin{center}\label{t:qodd}
\begin{table}[t]
\begin{tabular}{ |c|c| } 
 \hline
 Number of distinguished zeroes & Number of polynomials  \\ 
\hline
 $0$ & $q-1$ \\ 
\hline
 $q-3$ & $\frac{q(q-1)^2}{2}$ \\ 
\hline
$q-1$ & $\frac{q(q-1)(q+1)}{2}$ \\
\hline
$2(q-1)$ & $q(q-1)$\\
\hline
$q(q-1)$ & $1$ \\
 \hline
\end{tabular}
\caption{Number of polynomials with given number of distinguished zeroes when $q$ is odd}
\label{tab1}
\end{table}
\end{center}
We remark that, in the particular case when $q = 3$, then the nonzero constant polynomials as well as the polynomials satisfying the conditions $a_2 \neq 0$ and $\D$ a nonzero square in $\Fq$ admit no distinguished zeroes. We now study the case when $q$ is even. 
The proof is essentially similar, but the difference lies in the fact that every element of $\Fq$ is a square in $\Fq$. We include the complete proof for the ease of the reader. 

\begin{thm}\label{qeven}
Let $q \ge 4$ be even, $m=2$, and $f \in \Fq[x_1, x_2]$ be given by $f = a_0 + a_1 (x_1 + x_2) + a_2 x_1 x_2$. If $\D :=  a_1^2 - a_0 a_2$, then.
\begin{equation*}
|Z_{\Fq, D} (f)|=
\begin{cases}
0,  \ \ \ \text{if} \ \{a_0 \neq 0   \ \text{and} \ (a_1, a_2) = (0, 0)\} \\ 
 \ \ \ \ \ \ \ \ \ \ \text{or} \ \{a_1 \neq 0 \ \text{and} \ (a_0, a_2) = (0, 0)\} \\
q, \ \  \text{if} \   a_0a_1 \neq 0 \ \text{and} \  a_2 = 0  \\
q-2, \ \  \text{if} \  a_2 \neq 0,  \ \text{and} \ \D \neq 0 \\
2(q-1), \ \ \text{if} \ a_2 \neq 0 \ \text{and} \ \D = 0 \\
q(q-1), \ \ \text{if} \ (a_0, a_ 1, a_2) = (0, 0, 0)
\end{cases}
\end{equation*}
\end{thm}

\begin{proof}
If $f = 0$, then $|Z_{\Fq, D}(f)| = \P(q, 2)$. As in Proposition \ref{qodd}, it is clear from Corollary \ref{cor1} that if $|Z_{\Fq, D}(f)| = \P(q, 2)$, then $f=0$.  So we may assume that $f \neq 0$, i.e. $(a_0, a_1, a_2) \neq (0, 0, 0)$.   We again distinguish the proof into several cases:

\begin{enumerate}
\item[\textbf{Case 1:}] Suppose $a_2 = 0$. If $a_1 =0$, then $f$ is a nonzero constant polynomial which does not have any zeroes. So we may assume that $a_1 \neq 0$. 

\begin{enumerate}
\item[\textit{Subcase 1:}] If $a_0 \neq 0$, then all the zeroes of $a_0 + a_1(x_1 + x_2)$ are distinguished.  Note that if $a_0 + a_1 (x_1 + x_2)$ has a nondistinguished zero, say $(c, c)$, then $a_0 = a_1 (c + c) = 0$ since $q$ is even. This leads us to a contradiction to the assumption that $a_0 \neq 0$. Consequently, $|Z_{\Fq, D}(f)| = q$. 
\item[\textit{Subcase 2:}]
Then the zeroes of the polynomial $ a_1 (x_1 + x_2)$ are not distinguished.  Thus  $|Z_{\Fq, D}(f)| = 0$. 
\end{enumerate}

\item[\textbf{Case 2:}] Suppose $a_2 \neq 0$. As in Proposition \ref{qodd}, after a suitable change of coordinates, we get a polynomial 
$$f' (X_1, X_2) =  X_1 X_2 - \frac{a_0a_2 - a_1^2}{a_2^2},$$
with $|Z_{\Fq, D} (f)| = |Z_{\Fq, D} (f')|$. 

\begin{enumerate}
\item[\textit{Subcase 1:}] Suppose $\D = 0$. Then the polynomial $X_1X_2$ has exactly $2(q-1)$ many distnguished zeroes. 

\item[\textit{Subcase 2:}] Suppose $\D \neq 0$. Since $q$ is even, $D$ is a square in $\Fq$. Note that the polynomial $X_1 X_2 - \D/a_2^2$ has $q - 1$ zeroes and out of them only one is nondistinguished. Consequently, such a polynomial have $q-2$ distinguished zeroes.
\end{enumerate}
\end{enumerate}
This completes the proof. 
\end{proof}

\begin{center}\label{t:qeven}
\begin{table}[t] 
\begin{tabular}{ |c|c| } 
 \hline
 Number of distinguished zeroes & Number of polynomials  \\ 
\hline
\hline
 $0$ & $2(q-1)$ \\ 
\hline
 $q$ & $(q-1)^2$ \\ 
\hline
$q-2$ & $q(q-1)^2$ \\
\hline
$2(q-1)$ & $q(q-1)$\\
\hline
$q(q-1)$ & $1$ \\
 \hline
\end{tabular}
\caption{Number of polynomials with given number of distinguished zeroes when $q$ is even}
\label{tab2}
\end{table}
\end{center}

Again, it is not very difficult to compute the number of polynomials that admits a given number of distinguished zeroes in the case when $q$ is even. We omit the explicit computations, but present the data in Table \ref{tab2}.

\section{Reed-Muller type codes from symmetric polynomials}\label{2}
Throughout this section, we will denote by $\Fq$ a finite field with $q$ elements where $q$ is a power of a prime number. As in Section \ref{defs}, we denote by $\SS_m$ the vector space consisting of all symmetric polynomials as in \eqref{eff}. As noted before, $\SS_m$ is a vector space of dimension $m+1$ over $\Fq$. Let $n = \P (q, m)$. 

\begin{defn}\normalfont
We fix an ordering $\{P_1, \dots, P_n\}$ of elements in $\AA_D^m$. Define an evaluation map
$$\ev : \SS_m \to \Fq^n, \ \ \ \text{given \ by} \ \ \ f \mapsto (f(P_1), \dots, f(P_n)).$$  
It is readily seen that $\ev$ is a linear map and consequently the image, $\C_m$ of $\ev$ is a code. 
\end{defn}

We discuss some properties of this code in the following proposition:

\begin{prop}\label{bigcode}
If $m < q$, then the code $\C_m$ is a nondegenerate $[n, k, d]$ code, where $n = \P(q, m)$, $k = m+1$ and $d = (q - m) \P (q - 1, m - 1)$. Furthermore, the code $\C_m$ is generated by minimum weight codewords.  
\end{prop}

\begin{proof}
The statement on the length of the code is trivial, while the fact that the code is nondegenerate follows readily by observing that $\ev (1) = (1, \dots, 1) \in C_m$. To show that $C_m$ is of dimension $m+1$, it is enough to show that the map $\ev$ is injective. To this end, let $f \in \SS_m$ with $\ev (f) = (0, \dots, 0)$. Then $|Z_{\Fq, D}(f)| = \P (q, m)$. But from Corollary \ref{cor1}, we see that, if $f \neq 0$, then $|Z_{\Fq, D}(f)| \le m \P(q-1, m-1)$. Since $m < q$, we have $m \P(q-1, m-1) < \P (q, m)$. This implies $f=0$. Consequently, the map $\ev$ is injective. The assertion on the minimum distance follows from Corollary \ref{cor1} since
\begin{align*}
d& = \min_{f \in \sigma_m} |\{i  \mid f (P_i) \neq 0\}| \\
& = \P(q, m) - \max_{f \in \Sigma_m} |\{i \mid f(P_i) = 0\}| \\
& = \P (q, m) - m \P (q-1, m-1) \\
&=(q -m) \P(q-1, m-1).
\end{align*}

   Moreover, it is clear from the last assertion of Corollary \ref{cor1} that the minimum weight codewords of $\C_m$ are given by $\ev (f)$ where $f$ is a Type I polynomial. Thus, to show that $\C_m$ is generated by minimum weight codewords, it is now enough to prove that $\SS_m$ is spanned by a set of $m+1$ Type I polynomials. Since $m + 1 \le q$, we may choose $\a_1, \dots, \a_{m+1} \in \Fq$ that are distinct. For each $i = 1, \dots, m+1$, we define 
$$f_{i} = (x_1 + \a_i) \cdots (x_m + \a_i).$$
Since $\a_1, \dots, \a_{m+1}$ are distinct, it follows from the Vandermonde determinant formula that $f_1, \dots, f_{m+1}$ are linearly independent. Since $\dim_{\Fq}{\SS_m} = m+1$, they span the vector space $\SS_m$. This completes the proof. 
\end{proof}

\begin{rem}\normalfont \label{relative}
We note that the relative minimum distance of  $\C_m$ is the same as that of the generalized Reed-Muller codes of order $m$. 
\end{rem}

The code $\C_m$ is made by evaluating each of the functions in $\SS_m$ at the points of $\AA_D^m$.
But the points of $\AA_D^m$ constitute a disjoint union of $S_m$-orbits, each of cardinality $m!$, where the symmetric group $S_m$ in $m$ letters acts freely by permuting the coordinates. This motivates us in defining a code of smaller length, namely, by constructing a smaller evaluation set, say $R_D$, consisting of one point from each of the $S_m$ orbits mentioned above. Again we fix an ordering of the elements in the set $R_D$, say $Q_1, \dots, Q_N$, where $N = {q \choose m}$.  

We now consider the restriction of the evaluation map, still denoted by $\ev$:
$$\ev:\SS_m \rightarrow \Fq^N \ \ \ \text{given by} \ \ \ \ f \mapsto (f (Q_1), \dots, f(Q_N)).$$
 Let $\C'_m$ denote the image of $R_D$ under the map $\ev$. The following proposition follows readily from Proposition \ref{bigcode}. 

\begin{prop} \label{risone}
If $m < q$, then $\C'_m$ is a nondegenerate $[N, K, D]$ linear code where $N={q \choose m}$,  $K=m+1$ and  
$D={q \choose m}-{q-1 \choose m-1}$.
\end{prop}

\begin{proof}
The assertions on length and dimension is readily obtained as in the case with Proposition \ref{bigcode}. The assertion on minimum distance is deduced from Corollary \ref{cor1} and the fact that the weight of any codeword $\ev (f) \in \C'_m$ is given by ${q \choose m} - \frac{1}{m!}|Z_{\Fq, D} (f)|$. 
\end{proof}

\section{Generalized Hamming weights} 
Ever since their introduction by V. Wei in \cite{Wei}, the computation of generalized Hamming weights of several codes have been in the center of interest of many mathematicians and coding theorists. The study of generalized Hamming weights of several evaluation codes has paved the way for a lot of research articles such as \cite{BD, GL, HP} among others.  Before we proceed further, we recall the definition of generalized Hamming weights of linear codes. 

\begin{defn}\normalfont
Fix positive integers $n$ and $k$. 
\begin{enumerate}
\item[(a)] Let $W$ be a linear subspace of $\Fq^n$.  We define the support of $W$, denoted by $\supp (W)$, as follows:
$$\supp(W) = \{i \mid x_i \neq 0 \ \text{for some} \ (x_1, \dots, x_n) \in W\}.$$
\item[(b)] Given an $[n, k]$ code $C$ and a positive integer $1 \le r \le k$, we define the $r$-th generalized Hamming weight of $C$, denoted by $d_r(C)$, as
$$d_r (C) = \min \{|\supp (W) \mid W \ \text{is a subspace of C}, \dim W = r\}.$$

\end{enumerate}
\end{defn}

In this section, we derive some natural upper bounds on the generalized Hamming weights of the codes $\C_m$ and $\C'_m$. At the outset, we remark that it is enough to derive any parameters related to the Hamming weight of codewords for one of the codes. Since, the codes $\C_m$ are somewhat more natural to work with, we choose to restrict our attention to them.

\begin{prop} \label{first}
Fix positive integers $1 \le r < m + 1 \le q$ and denote by $d_r$ the $r$-th generalized Hamming weight of $\C_m$. We have
$$d_r \le \P(q, m)  - m! {q - r \choose m-r}.$$
\end{prop}
\begin{proof}
 Since $1 \le r \le q$, there exist distinct elements $b_1, \dots, b_r \in \Fq$. For $i = 1, \dots, r$, we consider the polynomials 
$$f_i := (x_1 - b_i) \cdots (x_m - b_i).$$
Note that $f_1, \dots, f_r$ are linearly independent and as a consequence $\ev (f_1), \dots, \ev (f_r)$ span an $r$ dimensional subspace, say $E_r$ of $\C_m$. It follows that
$$d_r \le |\supp (E_r)| = \P (q, m) - |Z_{S, D} (f_1, \dots, f_r)|,$$
where, as usual, for any subspace $V \subset \Fq^n$,  $$\supp (V) = \{i \mid \exists (a_1, \dots, a_n) \in V, a_i \neq 0\}.$$
Now, an element $(a_1, \dots, a_m) \in Z_{S, D} (f_1, \dots, f_r)$ if and only if for each $i = 1, \dots, r$, there exists $j_i \in \{1, \dots, m\}$ such that $a_{j_i} = b_i$. Hence the unordered $m$-tuple (set) $\{a_1,\cdots,a_m\}$ must be chosen in such a way that it contains $\{b_1,\cdots, b_r\}$. The remaining $m-r$ elements can be chosen arbitrarily among those $q-r$ elements not among the $b_i$. This gives ${q - r \choose m-r}$ unordered sets $\{a_1,\cdots,a_m\}$, and, at last, exactly $m!{q - r \choose m-r}$ ordered $m$-tuples contained  in $Z_{S, D} (f_1, \dots, f_r).$

\end{proof}

\begin{rem} \label{equalto}\normalfont
We note that the determination of the $r$-th generalized Hamming weight of $\C_m$ (resp. $\C'_m$) is equivalent to computing the maximum number of common zeroes of $r$ linearly independent elements of $\Sigma_m$ in $\A_D^m (\Fq)$ (resp. $R_D$). It follows trivially that $d_r(\C_m)=m!d_r(\C'_m)$. The following corollary is now immediate:
\end{rem}
 
\begin{cor} \label{second}
$d_r(\C'_m) \le {q \choose m}-{q-r \choose  m-r}.$
\end{cor}
 
The following proposition shows that the bounds obtained in Proposition \ref{first} is exact for the largest two values of $r$. 
\begin{prop} \label{uppercases}
We have 
\begin{enumerate}
\item[(a)] $d_{m+1}(\C_m)=m!d_{m+1}(\C'_m)=m!{q \choose m}$.
\item[(b)] $d_{m}(\C_m)=m!d_{m}(\C'_m)= m!\left({q \choose m}-1\right).$
\end{enumerate}
\end{prop}
\begin{proof}
Part (a) follows trivially since $(1, \dots, 1) \in \C_m$ and hence $\C_m$ is a nondegenerate code. We prove the part (b) for the code $\C'_m$. 

A generator matrix for $\C'_m$ is a parity check matrix for its dual code. Such a matrix $M=(m_{i,j})$ can be formed by setting $m_{i,j}=$ the value of $\sigma_{i-1}$ at point number $j$ in $\mathbb{R}_{\Fq}$, for some fixed order of the points in $\mathbb{R}_{\Fq}$.
Another way to put it is that $m_{i,j}=$ the value of $\sigma_{i-1}$ at a  chosen point in orbit number $j$ of $S_m$ in $\AA_D(\Fq)^m$, for some fixed order of the orbits in $\AA_D(\Fq)^m$.
Any two columns of this matrix are equal if and only if they are equal up to a non-zero, multiplicative constant. This is because their first entries are both $1(=\sigma_0).$ The last observation immediately shows that no column of $M$ is zero. Moreover any two columns are different.
This is because the elementary, symmetric functions $\sigma_1,\cdots,\sigma_m$ separate orbits of $S_m$ on $\AA_D(\Fq)^m$. (If 
$$X^m-\sigma_1 X^{m-1}+\cdots+(-1)^m\sigma_m=(X-\alpha_1)(X-\alpha_2)\cdots(X-\alpha_m),$$ 
then the $\alpha_i$ are unique up to order, since $\Fq[X]$ is a UFD). Hence no two columns are parallel vectors either (i.e. no two columns are equivalent up to a non-zero multiplicative constant). Hence the minimum distance of the dual code of $\C'_m$ is at least $3$. By Wei duality $$d_{m}(\C'_m)=\text{ length }(\C'_m)-1={q \choose m}-1.$$
This completes the proof.
\end{proof}
Propositions \ref{risone} and \ref{uppercases} give all $3$ generalized Hamming weights $d_r$ for $\C'_m$ and $\C_m$ for the  case $m=2$. If $m=q-1$, then 
$\C'_m$ fills the whole ambient space $\Fq^q$, and everything is trivial. It is a challenge, though, to give good results in the intermediate cases $3 \le m \le m-2$.

\section{The case $m=2$.} \label{mtwo}

As it is clear from the work done in previous sections, we are interested in computing the basic parameters such as length, dimension, minimum distance, generalized Hamming weights and the weight distributions for the codes $\C_m$ and $\C'_m$. In this section, we completely determine these parameters for the codes when $m = 2$. To begin with, we derive from Proposition \ref{bigcode} that $\C_m$ is an $[n, k, d]$ code, where 
$$n = q (q-1), \ \ \ k = 3, \ \ \  \text{and} \ \ \ d = (q-1)(q-2).$$ 
Furthermore, it follows from Propositions \ref{risone} and  \ref{uppercases} that 
$$(d_1, d_2,d_3)=\left((q-1)(q-2), q(q-1) - 2, q(q-1)\right),$$
where $d_1, d_2, d_3$ denote the first, second and third generalized Hamming weights for the code $\C_2$.  
We now proceed to determine the weight distribution for the code $\C_2$. To this end we introduce the following notation: 

\begin{defn}\normalfont
Let $w$ and $r$ be integers satisfying $0 \le w \le q(q-1)$ and $1 \le r \le 3$. Define
\begin{enumerate}
\item[(a)] $A_w :=$ the number of codewords of $\C_2$ of Hamming weight $w$. 
\item[(b)]  $A_w^{(r)}:=$ the number of $r$-dimensional subcodes of $\C_2$ of support weight $w$.
\end{enumerate}
\end{defn}

Let $c \in \C_2$ be a codeword. Then $c = \ev (f)$ for some $f \in \Sigma_2$. It follows that $c$ is a codeword of Hamming weight $w$ if and only if $|Z_{\Fq, D} (f)| = q(q-1) - w$. One can now readily compute the values of $A_w$ from Tables \ref{tab1} and \ref{tab2} for all values of $w$.  We have the following results:
  
\begin{prop} \label{oddcase}
If $q$ is odd, and $q\ge 5$, then we have
\begin{equation*}
A_w=
\begin{cases}
1,  \ \ \ \text{if} \ w = 0 \\
q(q-1), \ \  \text{if} \   w = (q-1)(q-2) \\
\frac{q(q-1)(q+1)}{2}, \ \  \text{if} \  w = q(q-1) - (q-1) \\
\frac{q(q-1)^2}{2}, \ \ \text{if} \ w = q(q-1) - (q-3) \\
(q-1), \ \ \text{if} \ w = q(q-1) \\
0, \ \  \  \ \ \text{otherwise}.
\end{cases}
\end{equation*}
\end{prop}
We remark that for $q=3$, we have $A_0 = 1,  A_2 = 6, \ A_4 = 12$ and $A_6 =8$.

\begin{prop} \label{evencase}
If $q$ is even, and $q \ge 4$, then we have
\begin{equation*}
A_w=
\begin{cases}
1,  \ \ \ \text{if} \ w = 0 \\
q(q-1), \ \  \text{if} \   w = (q-1)(q-2) \\
q(q-1)^2, \ \  \text{if} \  w = q(q-1) - (q-2) \\
(q-1)^2, \ \ \text{if} \ w = q(q-1) - 1 \\
2(q-1), \ \ \text{if} \ w = q(q-1) \\
0, \ \  \  \ \ \text{otherwise}.
\end{cases}
\end{equation*}
\end{prop} 

We now turn our attention towards computing $A_w^{(i)}$-s for all values of $1 \le w \le q(q-1)$ and $i = 1, 2, 3$ for the code $\C_2$. To this end, we have the following result:

\begin{prop}
For $1 \le w \le q(q-1)$ and $i = 1, 2, 3$ we have
\begin{equation*}
A_w^{(i)}=
\begin{cases}
\frac{A_w}{q-1},  \ \ \ \text{if} \ \ \ i =1 \\
\frac{q(q-1)}{2}, \ \  \text{if} \ \ \   w = q(q-1) - 2 \ \text{and} \ i=2 \\
\frac{q^2 + 3q + 2}{2}, \ \  \text{if} \  \ \ w = q(q-1) \ \text{and} \ i = 2 \\
1, \ \ \text{if} \ \ \ w = q(q-1) \ \text{and} \ i =3, \\
0, \ \ \text{otherwise}.
\end{cases}
\end{equation*}
\end{prop}

\begin{proof}
The assertions concerning the cases when $i =1$ and $i = 3$ are clear. To prove the claims concerning the cases when $i = 2$, we must analyze the possible number of distinguished points on the intersection of two curves given by $f_1, f_2 \in \Sigma_2$ such that $f_1$ and $f_2$ are linearly independent. Suppose that 
$$f_1 (x, y) = a_0 + a_1 (x + y) + a_2 xy \ \ \ \text{and} \ \ \ f_2 (x, y) = b_0 + b_1 (x + y) + b_2 xy.$$
We claim that $f_1$ and $f_2$ have no common factors. To see this, first note that, $f_1$ is not a nonzero constant multiple of $f_2$ since they are linearly independent. However, if $f_1$ has a factor of degree one, then $f_1 = c(x-a)(y-a)$ for some $a, c \in \Fq$. The fact that $f_1$ and $f_2$ have a common factor, now readily implies that $f_2 = d(x-a)(y-a)$ for some $d \in \Fq$. This is a contradiction. Now the projective closures of the zero sets $V(f_1)$ and $V(f_2)$ are given by homogeneous polynomials $F_1$ and $F_2$ of degree $2$, namely
$$F_1 = a_0 z^2 + a_1(x+y)z + a_2 xy \ \ \text{and} \ \ F_2 = b_0 z^2 + b_1(x+y)z + b_2 xy.$$
By Bezout's theorem, the projective curves given by $F_1$ and $F_2$ intersect at exactly $4$ points over the algebraic closure, counting multiplicities. We also observe that they have two points on the line $z=0$ in common, namely $[0:1:0]$ and $[1:0:0]$. Hence they have at most $2$ points in common in the affine space $\AA^2(\Fq).$
To this end, we observe that if $V(f_1)$ and $V(f_2)$ have points in common in $\AA^2_D(\Fq)$, then by symmetry, the points will be of the form $(\a,\beta)$ and $(\beta, \a)$ for some $\a, \beta \in \Fq$ with $\a \neq \beta$. 
Thus the affine curves $V(f_1)$ and $V(f_2)$ either do not intersect in $\AA^2_D(\Fq)$ or they intersect in exactly $2$ points. It is thus evident that $A_w^{(2)} = 0$ for all values of $w$ other that $q(q-1)$ and $q(q-1) - 2$. We now compute $A_w^{(2)}$ for $w = q(q-1)$ and $w = q(q-1) - 2$. It is enough to compute the same for $w = q(q-1) -2$.   
The elements of $\Sigma_2$ that contain the above points form a $2$ dimensional linear system of curves, which in turn gives us a two dimensional subcode that has weight $q(q-1) - 2$. On the other hand, there are $q(q-1)/2$ ways of choosing two such points from $\AA^2(\Fq)$. This shows that $A_{q(q-1) - 2}^{(2)} = q(q-1)/2$. Since there are a total of $q^2 + q + 1$ number of $2$ dimensional subcodes of $\C_2$, the assertion on $A_{q(q-1)}^{(2)}$ follows trivially. 
\end{proof}

Let $(\C_2)^{(s)}=\C_2 \otimes_{\mathbb{F}_q} \mathbb{F}_{q^s}$ for $s \geqslant 1$. It is a linear code over $\mathbb{F}_{Q}$, for $Q=q^s, $ with the same generator matrix as $\C_2$ itself.
In~\cite{J}, one gives a relation between the higher weight spectra of a linear code and the usual weight spectrum (only counting individual words of each weight), for such  an extension code over larger, finite fields. Denote the number of codewords of weight $w$ for $(\C_2)^{(s)}$
by  $P_w(Q).$ Then:

\[P_w(Q) = \sum_{r=0}^k A_w^{(r)} \prod_{i=0}^{r-1}(q^s-q^i)= \sum_{r=0}^k A_w^{(r)} \prod_{i=0}^{r-1}(Q-q^i).\]
This gives:
\begin{cor}
For $(\C_2)^{(s)}$ we have, if $q\ge 7$ is odd :
$$P_0(Q)=1,\textrm{ }P_{n-2(q-1)}(Q)=q(Q-1),\textrm{ }P_{n -(q-1)}(Q)=\frac{q^2+q}{2}(Q-1),\textrm{ }$$
$$P_{n-(q-3)}(Q)=\frac{q^2-q}{2}(Q-1),\textrm{ }P_{n-2}(Q)=\frac{q^2-q}{2}(Q-1)(Q-q),\textrm{ }$$
$$P_n(Q)=(Q-1)(Q^2+\frac{-q^2+q+2}{2}Q+\frac{q^3-3q^2-2q+2}{2}).$$
\end{cor}
 
We leave it to the reader to find analogous formulas for $q=3,5$, and for even $q \ge 4$.

\section{The case $m=3$.} \label{mthree}
If $m=3$, the only unknown generalized Hamming weight of $\C_m=\C_3$ is $d_2$, since we know that $d_1=6{q \choose\ 3} - 6{q-1 \choose 2}, d_3=6{q \choose 3}-6, \text{ and } d_4=6{q \choose 3},$ by the results above.
Furthermore we know from Proposition \ref{basic}, as for general $m$,  that there are precisely $q(q-1)$ codewords of minimal weight, namely ev$(c(x_1-b)\cdots(x_m-b))$, for $q$ choices of $b$, and $q-1$ choices of non-zero $c$. 

The case $q=4$ is a special case of the trivial case $m=q-1$ mentioned at the end of Section 4, with $d_r(\C_3)=6r$ for all $r$.

In order to illustrate the complexity, we give our only example below, of a more non-trivial result for $m=3$. We leave it to further research to find good results for $m \ge 3$ in general.

\begin{example}
{\rm If $q=5$, one can show that $d_2(\C_3)=42$ 
(and hence $(d_1,d_2,d_3,d_4)=(24,42,54,60)$).

For practical calculations in order to show this it will be more convenient to work with the smaller code $\C'_3$
and  use the fact that 
$d_r(\C_r)=m!d_r(\C'_r)$, for all 
$m,r$ in question. Thus we will prove the equivalent statement $d_2(\C'_3)=\frac{42}{6}=7$:
The word length of $\C'_3$ is $n=\frac{60}{6}$, and a generator matrix $M=(m_{i,j})$ of $\C'_3$ with $m_{i,j}= \sigma_{3}^{i-1}(P_j)$ becomes
\setcounter{MaxMatrixCols}{10}
\[\begin{bmatrix} 
1 & 1 & 1 & 1 & 1 & 1 & 1 & 1 & 1 & 1 \\
3 & 4 & 0 & 0 & 1 & 2 & 1 & 2 & 3 & 4 \\
2 & 3 & 4 & 1 & 3 & 2 & 1 & 4 & 4 & 1 \\
0 & 0 & 0 & 0 & 0 & 0 & 1 & 3 & 2 & 4 

\end{bmatrix}\]
Here the $P_j=\{a_j,b_j,c_j\} \subset \{0,1,2,3,4\}$ are ordered lexicographically with respect to the ordering of ${\mathbb F}_5$, and they represent the $10=\frac{60}{6}$ orbits of 
$S_3$ in $\mathbb{A}_D(\mathbb{F}_5)^3$. 
Hence $P_1=(0,1,2),P_2=(0,1,3),\cdots,P_{10}=(2,3,4)$, and the first column of $M$ becomes the transpose of:

$(\sigma_3^0(P_1),\sigma_3^1(P_1),\sigma_3^2(P_1), \sigma_3^3(P_1))=
(1,0+1+2,0\cdot 1+0 \cdot 2+1 \cdot 2, 0 \cdot 1 \cdot 2)=(1,3,2,0).$ The nine other columns are computed in the same way.

The columns of $M$ may be interpreted as points of projective $3$-space $Proj({\mathbb F}_5[y_0,y_1,y_2,y_3])=\mathbb{P}^3$ over ${\mathbb F}_5$. One then has the well known result (see for example \cite[p. 276]{HTV}):

\begin{lemma} \label{well}
If the columns of a $(k \times n)$-generator matrix of a linear code $C$ are interpreted as points of a projective space $\mathbb{P}^{k-1}$, then for all $r=1,\cdots,k$ we have 

$d_r(C)=n - m_r$, where $m_r$ is the maximal number of column points that are contained in a codimension $r$-space in $\mathbb{P}^{k-1}$.
\end{lemma}

One then understands that $d_1(\C'_3)=10$ - the maximal number of points from $P_1,\cdots,P_{10}$ that are contained in a (projective) plane $=10-6=4$, and $d_3(\C'_3)=10$ - the maximal number of points from $P_1,\cdots,P_{10}$ that are contained in a point $=10-1=9$ (obvious), since one already knows that $d_1(\C'_3)=\frac{24}{6}=4,$ and $d_3(\C'_3)=\frac{54}{6}=9.$ This can also be checked by inspecting the matrix $M$ directly. 

Our only new task is to show that the maximal number of points from $P_1,\cdots,P_{10}$ that are contained in a (projective) line is $3$, so that $d_2(\C'_3)=10-3=7$. 

Our first observation is that the first $3$ columns of $M$ are dependent (compute its upper left $3$-minor), so the maximal number of column points on a line is at least $3$. We will show that it is at most $3$.

For $q=5$ there are $q^3+q^2+q+1=156$ planes in $\mathbb{P}^3$, and among them we have the $5$ planes $W_j=0,$ for $j=0,\cdots,4$ that correspond to (totally) reducible elements 
$$ \prod_{i=1}^3 (x_i - j)$$ of $\SS_5$, and each of them contains ${q-1 \choose m-1}=6$ of the $10$ points. We see directly from the matrix description that $W_0=y_3=0$ is one of them.
By the last part of Theorem \ref{basic} all other planes contain at most $5$ column points from $M$. 
A computer analysis, or a by-hand calculation, shows
that there is no plane (corresponding to irreducible elements of $\SS_m$ that contain exactly $5$ column points (of $M$), although the last part of Theorem \ref{basic}, which is obviously not sharp in this case, allows it. Moreover there are exactly ${5 \choose\ 2}=10$ planes that contain exactly $4$ of  the $10$ points (For each unordered pair $W_i,W_j$, with $i \ne j$, there is exactly one such plane $V_{i,j}=0$ of type $=W_i+g_{i,j}W_j=0$, with $g_{i,j} \ne 0$).  

We will use this to prove that no more than $3$ of the (column) points (of $M$) are on a line. 
First, two planes $W_i=0$ and $W_j=0$, with $i \ne j$, are well known to intersect in exactly $3$ points, by the case $r=2$ of the proof of Proposition \ref{first}. 

If two distinct planes of type $V_{i,j}=0$ contained the same $4$ points, add any point outside $L$, which is the intersection of those two planes. Then those $5$ points would span a plane. This plane would have to be one of the $W_j=0$, since no other planes contains at least $5$ points, as we have seen.  But a plane of type $W_j=0$ does not contain $4$ points on a line, as we will now show, so that is impossible. 

The assertion that the plane of type $W_0=y_3=0$ does not contain $4$ points on a line, follows by direct inspection of the $6$ leftmost columns of $M$. One finds that the only dependent column triples are $(123), (145), (246), (356)$, and since no column quadruple contains two of these triples, we cannot have $4$ column points on a (projective) line.

Let $l \in \{1,2,3,4\}$. If instead we choose a generator matrix, with $m_{i,j}=\sigma_m^{i-1}(P_j-(l,l,l)))$, and order the $P_j$ lexicographically as 
$(l,l+1,l+2),(l,l+1,1+3),\cdots,(l-3,l-2,l-1),$
then this "new" matrix will in fact be equal to $M$. But now $y_3=0$ corresponds to the totally reducible element 
$$ \prod_{i=1}^3 (x_i - l)$$ of $\SS_5.$
and therefore there are not $4$ collinear column points in this plane, either,
for $l=1,2,3,4$. But in the initial matrix description this corresponds to the plane $W_l=0$,
and linear relations between columns do not change
when one performs row operations on a matrix.
Hence no plane $W_l$ contains $4$ collinear points.

If a plane of type $V_{i,j}=0$ and a plane of type $W_l=0$ contained $4$ common points (in their intersection, a line $L$), we again have four points in a plane of type $W_l=0$, a contradiction.

Hence the maximal number of column points of $M$ on a line is at most $3$ also, and hence $d_2(\C'_3) = {q \choose\ 3} -
{q-2 \choose\ 3-2} =10-3=7$, for $q=5$. As a byproduct of this analysis we observe that in addition to $q(q-1)=20$ words of minimal weight $4$ there are exactly 
$(q-1){q \choose\ 2}=40$ codewords of "subminimal" weight, in this case $6$, in $\C'_3$ for $q=5$.}
\end{example} 




\setcounter{MaxMatrixCols}{6}



\textbf{Acknowledgment:} The authors sincerely thank the anonymous referees for their careful reading the initial version of this article and providing us with some important suggestions towards improving the article. We also thank Hiram Lopez for their comments and suggestions on this article.

\vspace{.5cm}


\begin{thebibliography}{A}


\bibitem{BD} P. Beelen and M. Datta, \textit{Generalized Hamming weights of affine Cartesian codes}, Finite Fields Appl. 51 (2018), pp. 130–145.

\bibitem{BDH} P. Beelen, M.  Datta and  M. Homma \textit{A proof of S{\o}rensen’s conjecture on Hermitian surfaces},
Proc. Amer. Math. Soc. 149 (2021), pp. 1431-1441

\bibitem {BS} P. Beelen and P. Singh \textit{Linear codes associated to skew-symmetric determinantal varieties},
Finite Fields and Their Applications
58, July 2019, pp. 32-45.

\bibitem {G} V. D. Goppa \textit{Geometry and Codes}, ISBN: 978-94-015-6870-8, Springer Verlag, 1988.

\bibitem{GL} S. R. Ghorpade and G.  Lachaud, \textit{Higher weights of Grassmann codes}, in Coding theory, cryptography and related areas (Guanajuato, 1998), pp. 122–131, Springer, Berlin, 2000.

\bibitem{HP} P. Heijnen and R. Pellikaan, \textit{Generalized Hamming weights of $q$-ary Reed-Muller codes}, IEEE Trans. Inform. Theory 44 (1998), no. 1, pp. 181–196.

\bibitem{HTV} J. W. P. Hirschfeld, M. A. Tsfasman, and S. G. Vladut, \textit{The Weight Hierarchy of Higher Dimensional Hermitian Codes},IEEE Trans. Inform. Theory,  40 (1994) no. 1, pp. 275-278.

\bibitem {J} R. Jurrius,  \textit{Weight enumeration of codes from finite spaces}, 
Des. Codes Cryptogr. 63, pp. 321-330, 2012.

\bibitem {N} D. Nogin, \textit{The minimum weight of the Grassmann codes $C(k,n$)}, Discr. Appl. Math.
28 (1990), pp. 149–156

\bibitem{W} E. Weiss, \textit{Generalized Reed-Muller Codes}, Information and Control 5, (1962) pp.213-222 

\bibitem{Wei} V. Wei,  \textit{Generalized Hamming weights for linear codes},  IEEE Trans. Inform. Theory 37 (1991), no. 5, pp. 1412 - 1418.
\end{thebibliography}
\end{document}